\documentclass[aip,graphicx]{revtex4-1}

\usepackage[utf8]{inputenc}
\usepackage[T1]{fontenc}
\usepackage{amsmath}
\usepackage[a4paper,left=2.cm,right=2.cm,top=2.5cm,bottom=2.5cm,includefoot=false,includehead=false]{geometry}
\usepackage{graphicx}
\usepackage{float}
\usepackage{indentfirst}
\usepackage{mathtools}

\usepackage{comment}
\usepackage{braket}
\usepackage{xcolor}

\usepackage{amsmath, amstext, amsopn, amsfonts, amsthm, amsbsy, latexsym, amssymb,textcomp,}

\newtheorem{lem}{Lemma}

\newtheorem{prop}{Proposition}

\newtheorem{theorem}{Theorem}
\theoremstyle{definition}
\newtheorem{ex}{Example}
\newtheorem{corollary}{Corollary}
\newtheorem{remark}{Remark}

\newtheorem{definition}{Definition}

\newcommand{\macierze}{\mathcal{M}_n(\mathbb{C})}

\draft 
\begin{document}

\title{Application of Shemesh theorem to quantum channels}

\author{Michał Białończyk}
\affiliation{Institute of Physics, Jagiellonian University, ul. Stanis\l{}awa \L{}ojasiewicza 11, 30-348 Krak\'o{}w, Poland}

\author{Andrzej Jamiołkowski}
\affiliation{Institute of Physics, Nicolaus Copernicus University, ul. Grudziądzka 5, 87-100 Toru{\'n}, Poland}
\author{Karol Życzkowski}
\affiliation{Center for Theoretical Physics, Polish Academy of  Sciences, Aleja
Lotników 32/46, PL-02-668 Warsaw, Poland}
\affiliation{Institute of Physics, Jagiellonian University, ul. Stanis\l{}awa \L{}ojasiewicza 11, 30-348 Krak\'o{}w, Poland}

\date{ July 11, 2018}

\begin{abstract}
Completely positive maps are useful in modeling the discrete evolution of quantum systems. Spectral properties of operators associated with such maps are relevant for determining the asymptotic dynamics of quantum systems subjected to multiple interactions described by the same quantum channel. We discuss a connection between the properties of the peripheral spectrum of completely positive and trace preserving map and the algebra generated by its Kraus operators $\mathcal{A}(A_1,\ldots A_K)$. By applying the Shemesh and Amitsur - Levitzki theorems to analyse the structure of the algebra $\mathcal{A}(A_1,\ldots A_K)$ one can predict the asymptotic dynamics for a class of operations.
\end{abstract}
\maketitle
\section{Introduction}
Quantum operations are fundamental mathematical objects used in description of the evolution of Quantum Opens Systems. Paradigmatic theoretical setup is the following: suppose one has the small system of interest modeled by a Hilbert space $\mathcal{H}$ and the environment, often called \emph{ancilla}, which is supported on a Hilbert space $\mathcal{K}$. If the initial states of the system and ancilla are $\rho \in \mathcal{B}(\mathcal{H})$ and $\xi \in \mathcal{B}(\mathcal{K})$ respectively, then the state of the system evolves according to:
\begin{equation}\label{eq:general_evolution}
\rho \, \rightarrow \rho' = \mathrm{Tr}_{\mathcal{K}}[U(\rho \otimes \xi)U^{\dagger}],
\end{equation}
where $U$ is a unitary that describes evolution of the total system, which is assumed to be closed and satisfies Schroedinger equation. Equation \ref{eq:general_evolution} defines a trace preserving completely positive map (\emph{quantum channel}) $\Phi \colon \mathcal{B}(\mathcal{H}) \rightarrow \mathcal{B}(\mathcal{H})$ which describes general discrete evolution of quantum open system \cite{holevo, lecture_notes}. Continuous evolution of such system is given by a family of completely positive maps $\lbrace \Phi \rbrace_{t \geq 0}$ which leads to a concept of Quantum Dynamical Semigroups \cite{alicki} and GKSL equation \cite{kossakowski, lindblad,CP17}.

On the other hand, instead of continuous time evolution, one can consider multiple application of the same channel:
\begin{equation}
\rho \, \rightarrow \Phi(\rho) \, \rightarrow \Phi^2(\rho)  \, \rightarrow \ldots \rightarrow \Phi^k(\rho) \, \rightarrow \ldots.
\end{equation}
Such evolutions are referred to as \emph{Quantum Markov Chains} \cite{accardi, szehr} or \emph{Repeated Quantum Interactions} \cite{nechita} and have been studied in the context of quantum networks and sequential measurements \cite{seq}. In analysis of such systems one is often interested in \emph{asymptotic dynamics}, that is behavior of the 
system after many applications of the same channel. 
Asymptotic dynamics is governed by the spectrum of the superoperator  $\Phi$ associated with the channel \cite{jex, nechita, chiribella}.
In particular one can show, that the dynamics becomes confined to the subspace spanned by the eigenvectors corresponding to unimodular eigenvalues, 
which are called the \emph{peripheral spectrum} of $\Phi$. 
Recent paper by  Novotny, Alber and Jex \cite{jex}  provides expressions allowing one to determine the
subspace into which the asymptotic dynamics is confined, also called {\sl attractor space}.
In this approach one assumes that the peripheral spectrum is known,
but for given quantum map 
the spectrum of the corresponding superoperator
can be difficult to determine.

Therefore it would be useful to have some \emph{effective} criteria to narrow the set of possible elements of the peripheral spectrum. By an effective criterion (or effective condition) we mean any procedure employing only finite
number of arithmetic operations. We emphasize that in applications of mathematics to physics it is often crucial to have some effective (computable) criteria then purely theoretical ones.
The main goal of this work is to provide operational criteria applicable to a wide class of quantum channels
allowing us to describe the peripheral spectrum
and to detect possible cycles appearing in asymptotic dynamics.



\subsection{Preliminaries and notation}

Let $\mathcal{M}_n(\mathbb{C})$ be the space of $n\times n $ complex matrices. We denote by $\mathcal{M}^{+,1}_n(\mathbb{C})$ the set of density matrices.  The general Markovian evolution of a quantum system can be described by a completely positive trace preserving (CPTP) map $\Phi \:\colon \macierze \rightarrow \macierze$ restricted to density matrices. It can be shown \cite{choi} that any completely positive (CP) map can be expressed in the form: 
\begin{equation}\label{eq:kraus_form}
\Phi(X) = \sum_{i=1}^K A_i X A_i^{\dagger},
\end{equation}
where $K \leq n^2$ and $A_i$ are $n \times n$ matrices called \emph{Kraus operators}. The map of the above form is trace preserving, if
\begin{equation}
\sum_{i=1}^K A_i^{\dagger} A_i = \mathbb{I}.
\end{equation}
Completely positive and trace preserving maps are called \emph{quantum operations} or \emph{quantum channels}. Furthermore, if $\Phi (\mathbb{I}) = \mathbb{I}$, that is,  
\begin{equation}
\sum_{i=1}^K A_i A_i^{\dagger} = \mathbb{I},
\end{equation}
then such $\Phi$ is called \emph{unital}. A trace preserving and unital map is also refered to as bistochastic map. 

For a map $\Phi$ in the form (\ref{eq:kraus_form}) the  \emph{dual map} is defined as
\begin{equation}\label{def:dual}
\Phi^{\circ}(X) = \sum_{i=1}^K A_i^{\dagger} X A_i.
\end{equation}
Furthermore, operators corresponding to maps $\Phi$ and $\Phi^{\circ}$ have the same spectrum (counting multiplicities).

Examining the form (\ref{eq:kraus_form}) we see that every CP map can be associated with the subalgebra of $\macierze$, namely the algebra $\mathcal{A} (A_1,\ldots A_K)$ generated by $A_1, \ldots A_K$. Intuitively, this algebra contains all expressions in the form $A_{i_1}^{n_1} A_{i_2}^{n_2} \ldots A_{i_K}^{n_K}$ and their linear combinations. This algebra is independent of the particular representation (\ref{eq:kraus_form}), and thus we will also use the notation $\mathcal{A}(\Phi)$. 

Every quantum map has to be positive , meaning that it preserves the positive cone $\mathcal{M}_n^+(\mathbb{C})$ in $\macierze$. Therefore, similarly to the case of the positive maps one can introduce the notion of \emph{irreducibility} :

\begin{definition}
A completely map $\Phi \:\colon \macierze \rightarrow \macierze$ is called \emph{irreducible} if there exists no nontrivial face of the cone $\mathcal{M}_n^+(\mathbb{C})$ invariant under the map $\Phi$.  
\end{definition}  

If the map $\Phi$ is given by its Kraus decomposition the above definition can be expressed in an equivalent way \cite{farenick} :

\begin{theorem}[Farenick]\label{thm:farenick}
A completely positive map $\Phi \: \colon \macierze \rightarrow \macierze$ in the form (\ref{eq:kraus_form}) is irreducible if and only if the operators $A_1,\ldots A_K$ have no nontrivial common invariant subspace (the trivial subspaces are $\lbrace 0 \rbrace$ and $\mathbb{C}^n$). 
\end{theorem} 

From Burnside theorem \cite{burnside}, a given subalgebra of $\mathcal{A} \subset \macierze$ has no nontrivial invariant subspace if and only if 
$\mathcal{A} = \macierze$. Therefore, Theorem \ref{thm:farenick} can be stated as follows : Quantum map in the form (\ref{eq:kraus_form}) is irreducible if 
and only if
\begin{equation}
\mathcal{A}(A_1,\ldots A_K) = \mathcal{M}_n(\mathbb{C}).
\end{equation} 

{

Note that the eigenvalues and eigenvectors of a complex $n\times n$ matrix cannot, in general, be computed by radicals when $n>4$ (the Abel--Ruffini theorem).
However, much of the spectral information can be obtained without radicals by using finite sequences of arithmetic operations.
For example, one can compute the minimal and characteristic polynomials in a finite number of steps.
This gives effective answers to some questions connected with the spectral structure of a given matrix.
In particular, we can check:
\begin{enumerate}
 \item  For a given matrix $A$, does $A$ have eigenvalues with multiplicity 2 or more?
 The answer is positive if the discriminant of the characteristic polynomial is equal to zero.
 \item A matrix $A$ can be represented in diagonal form if the discriminant of the minimal polynomial is equal to zero.
 \item Some invariant subspaces can be constructed in explicit form. 
 Namely, the so-called Krylov subspaces can be generated if the degree of the minimal polynomial is smaller than degree of the characteristic polynomial.
\end{enumerate}
}


The Perron-Frobenius theorem \cite{evans_krohn} states that, for a positive, unital and irreducible map $\Phi \: \colon 
\mathcal{M}_n(\mathbb{C}) \rightarrow \mathcal{M}_n(\mathbb{C})$, there exists eigenvalue $1$ which is nondegenerate, and corresponding eigenvector is strictly positive.  For completely positive, 
trace preserving or unital map one can show even more. Assuming that the \emph{peripheral spectrum} $\mathrm{spec}_1 (\Phi)$ of a given trace  preserving or unital 
positive map $\Phi \:\colon \macierze \rightarrow \macierze$ is defined as the set of eigenvalues of $\Phi$ with modulus 1, the following theorem holds 
\cite{groh} :

\begin{theorem}[Groh]\label{thm:Groh}
Let $\Phi \, \colon \mathcal{M}_n(\mathbb{C})\rightarrow \mathcal{M}_n(\mathbb{C})$ be a trace preserving or unital CP map. If $\Phi$ is irreducible, then there exists $m\in\lbrace 1,\ldots, n^2 \rbrace$ such that the peripheral spectrum of $\Phi$ has the form:
\begin{equation}
\mathrm{spec}_1(\Phi) = \lbrace e^{\frac{2\pi i}{m}k} ,\, k=1,\ldots, m \rbrace.
\end{equation}
\end{theorem}
It is worth noting, that complete positivity is not necessary  -  it is sufficient to assume that a positive map $\Phi$ is trace preserving or unital and for every $X \in \mathcal{M}_n(\mathbb{C})$ satisfies the so called \emph{Schwarz inequality}\cite{evans_krohn} : $\Phi(X)^{\dagger} \Phi(X) \leq \Phi(X^{\dagger} X)$ .

According to Corrolary 7 from Ref. \cite{chiribella} one can further restrict the possible values in peripheral spectrum of an irreducible map if linear span of Kraus operators contains an invertible element:

\begin{theorem}\label{thm:stronger}
Let $\Phi \, \colon \mathcal{M}_n(\mathbb{C})\rightarrow \mathcal{M}_n(\mathbb{C})$ be a trace preserving or unital CP map with Kraus operators $\lbrace A_i \rbrace_{i=1}^K$. If $\Phi$ is irreducible and $\mathrm{span} \lbrace A_i \rbrace_{i=1}^K$ contains an invertible element, then every peripheral eigenvalue $\lambda$ of $\Phi$ satisfies $\lambda^{n} = 1$.  
\end{theorem} 
Note that assuming only irreducibility, one can only predict that the peripheral spectrum is a cyclic subgroup of $U(1)$ with order at most $n^2$. Theorem \ref{thm:stronger} tells, that if additionally $\mathrm{span} \, \lbrace A_i \rbrace_{i=1}^K$ (which is true for generic map $\Phi$) then one can restrict the order to the divisors of $n$, which is significant improvement.  

If a positive, trace  preserving or unital map $\Phi$ has no eigenvalues different then $1$ in its peripheral spectrum, then such map is called \emph{primitive}. Reference \cite{wieleandt} provides a useful, operational criterion for primitivity of a CP map:
\begin{theorem}
\label{thm:primitive}
For a CP map $\Phi(X) = \sum_{i=1}^K A_i X A_i^{\dagger}$ define $S_m(\Phi)$ as the linear span of all possible products of exactly $m$ Kraus operators:
\begin{equation}
S_m(\Phi) = \mathrm{span}\lbrace A_{i_1}\ldots A_{i_m}, \quad 1 \leq i_1, \ldots i_m \leq K \rbrace. 
\end{equation}
Then $\Phi$ is primitive if there exists $m\geq 1$ such that:
\begin{equation}
S_m(\Phi) = \mathcal{M}_n(\mathbb{C}).
\end{equation}
\end{theorem}
The main objective of this paper is to examine the possibilities of weakening the irreducibility requirement while still preserving the spectral properties similar to the conclusion of Theorem 2. We present an approach based on the investigation of the structure of $\mathcal{A}(\Phi)$ to show  the connection with  spectral properties of the corresponding operator $\Phi$. 

We will make use of the classical result stating that if a given algebra $\mathcal{A}(\Phi)$ is $^{\star}$-algebra i.e. it is closed under Hermitian conjugation, then one can choose an orthonormal basis in which $\mathcal{A}$ is block - diagonal \cite{noncommutative} :  

\begin{corollary} \label{wn:dekompozycja_unitalnych}
Let  $\mathcal{H}\simeq \mathbb{C}^n$, $\mathcal{B}(\mathcal{H})\simeq \mathcal{M}_n(\mathbb{C})$ and let $\Phi \, \colon \mathcal{M}_n(\mathbb{C}) \rightarrow \mathcal{M}_n(\mathbb{C})$ be a trace preserving or unital CP map written in the Kraus form,  
\begin{equation}
\Phi(X) = \sum_{i=1}^K A_i X A_i^{\dagger}.
\end{equation}
If $\mathcal{A}(A_1,\ldots A_K)$ is a $^{\star}$-algebra, then there exists an orthonormal basis $\lbrace e_i \rbrace_{i=1}^n$ and natural numbers $d_1,\ldots d_N$ for which all Kraus operators have the form : 
\begin{equation}\label{eq:dekompozycja_unitalnych}
A_i = 
\begin{bmatrix}
A_{i1} & 0 & \hdots & 0 \\
0      &    A_{i2} & \hdots & 0\\
\vdots & \hdots & \ddots & 0 \\
0 & \ldots & 0 & A_{iN} \\
\end{bmatrix}
\end{equation}
where each block $A_{im}$ has dimension $d_m\times d_m$, $\sum_{j=1}^N d_j=n$ and $\mathcal{A}(A_{1m},\ldots, A_{Km}) \simeq \mathcal{M}_{d_m}(\mathbb{C})$. Consequently, there exists a decomposition of the Hilbert space $\mathcal{H}$ such that
\begin{equation}\label{eq:dekompozycja_hilberta}
\mathcal{H} = \bigoplus_{j=1}^N \mathcal{H}_j,
\end{equation}
where $\mathrm{dim} \, \mathcal{H}_k = d_k$ for $k=1, \ldots N$, $A_{im} \colon \, \mathcal{H}_m \rightarrow \mathcal{H}_m$ and 
\begin{equation}
A_{im} = A_i|_{\mathcal{H}_m}.
\end{equation}
\end{corollary} 

The most important examples of the maps for which the algebra $\mathcal{A}(\Phi)$ is a $^{\star}$-algebra are unital quantum channels \cite{kribs}. In essence, the above Theorem states that if $\mathcal{A}(\Phi)$ is a $^{\star}$-algebra, then one can decompose (,,reduce'') the Kraus operators into smaller, 
irreducible blocks.

\section{Eigenspaces corresponding to unimodular eigenvalues}

The problems related to the structure of peripheral spectra, eigenvectors corresponding to unimodular eigenvalues and fixed points of CP maps (and, more generally, positive maps) have been extensively studied in literature \cite{fixed_points, lecture_notes, jex, inverse_eigenvalue, gluck, multiplicative}. In particular, the problem of finding sufficient conditions for cyclicity of the peripheral spectrum was adressed in Ref. \cite{gluck, lotz}. In this paper we attempt to give efective criteria for cyclicity of peripheral spectrum for certain quantum channels given by their Kraus decomposition. As the starting point for further considerations, in Theorem \ref{thm:summary} we present selected known results, focusing on the role of algebraic structure. We include a proof for completeness.   

For a given CP map $\Phi$ we use the notion of 
the set $\chi (\Phi)$ of unimodular eigenmatrices and
the fix point space $\chi_F (\Phi)$,

\begin{equation}
\chi (\Phi)  = \lbrace X \in \mathcal{M}_n(\mathbb{C}) \colon \, \exists \lambda \in \mathbb{C}, \, \Phi(X) = \lambda X, \, \vert \lambda \vert = 1 \rbrace,  
\end{equation}
\begin{equation}
\chi_F (\Phi)  = \lbrace X \in \mathcal{M}_n(\mathbb{C}) \colon \, \Phi(X) = X \rbrace.   
\end{equation}  

\begin{theorem}\label{thm:summary}
Let $\Phi$ be a unital, completely positive map given by its Kraus decomposition $\Phi(X) =  \sum_{i=1}^K A_i X A_i^{\dagger}$. Then, the following conditions are equivalent:
\begin{enumerate}
\item Dual map $\Phi^{\circ}$ has full rank fixed point, that is there exists $\rho > 0$ such that $\Phi^{\circ}(\rho) = \rho$.
\item The spaces $\mathrm{span}\, \chi(\Phi)$ and $\chi_F(\Phi)$ are $\star$-algebras. 
\item The following equalities hold:
\begin{equation}
\chi(\Phi) = \lbrace X \in \mathcal{M}_n(\mathbb{C}) \colon \, \exists \lambda \in \mathbb{C} \, \forall i \colon \, A_i X - \lambda X A_i = 0, \, \vert \lambda \vert  =1 \rbrace.  
\end{equation}
\begin{equation}\label{eq:expr_for_fixed}
\chi_F(\Phi) = \lbrace X \in \mathcal{M}_n(\mathbb{C}) \colon \,  \forall i \colon \, [X,A_i] = [X,A_i^{\dagger}] = 0 \rbrace. 
\end{equation}
\item The Kraus algebra $\mathcal{A}(\Phi)$ is a $\star$-algebra. 

\end{enumerate}

\end{theorem}
\begin{proof}

$1 \Rightarrow 2$. Let $\rho>0$ be full rank fixed point of a map $\Phi^{\circ}$ and $A$ be an eigenvector to some peripheral eigenvalue $\lambda$. Then $A^{\dagger}$ is an eigenvector to eigenvalue $\bar{\lambda}$ and one has $\Phi(A^{\dagger}) \Phi(A) = A^{\dagger} A$. On the other hand, it holds $\mathrm{Tr}\,[\Phi(A^{\dagger} A) \, \rho] = \mathrm{Tr} \, [A^{\dagger} A \,  \Phi^{\circ} (\rho)] = \mathrm{Tr} \, [A^{\dagger} A \, \rho]$.  Putting this together, one obtains:
\begin{equation}
\mathrm{Tr} [(\Phi(A^{\dagger} A) - \Phi(A^{\dagger})\, \Phi(A))\rho] = 0. 
\end{equation}
Now, from Schwarz inequality $\Phi(A^{\dagger} A) - \Phi(A^{\dagger}) \Phi(A) \geq 0$ and because $\rho > 0$ it must be $\Phi(A^{\dagger} A) = \Phi(A^{\dagger}) \Phi(A)$. Equality in Schwarz inequality implies \cite{lecture_notes, paulsen}, that for every matrix $X$ one has:
\begin{equation}
\Phi(A^{\dagger} X) = \Phi(A^{\dagger} ) \Phi(X).
\end{equation} 
Taking arbitrary $A^{\dagger}, X \in \chi(\Phi)$, from above equation one gets $A^{\dagger} X \in \chi(\Phi)$. Therefore $\mathrm{span} \, \chi(\Phi)$ is closed with respect to multiplication and hermitian conjugation and, in conseqence is a $\star$-algebra. Similar reasoning leads to conclusion, that $ \chi_F(\Phi)$ is a $\star$-algebra. 

$2 \Rightarrow 3$. Let $\Phi(X) = \lambda X$ for $\vert \lambda \vert = 1$. Then $\Phi(X^{\dagger}) = \bar{\lambda} X^{\dagger}$. Assuming that $\chi(\Phi)$ is closed with respect to multiplication and using positivity of $\Phi$ one obtains $\Phi(X X^{\dagger}) = X X^{\dagger}$. Now,
\begin{equation}
\label{eq:suma_dodatnich}
0 \leq \sum_{i=1}^K (A_i X - \lambda X A_i )(A_i X - \lambda X A_i)^{\dagger} =  \sum_{i=1}^K A_i X X^{\dagger} A_i^{\dagger} - \vert \lambda \vert^2 X X^{\dagger} - \vert \lambda \vert^2 XX^{\dagger} + \vert \lambda \vert^2 X X^{\dagger} = 0. 
\end{equation}  
Because each term $(A_i X - \lambda X A_i )(A_i X - \lambda X A_i)^{\dagger}$ is positive-semidefinite, equation (\ref{eq:suma_dodatnich}) can be satisfied only if $A_i X - \lambda X A_i = 0$ far all $i$. The same reasoning leads to the expression for $\chi_F(\Phi)$. 

$3 \Rightarrow 4$. From equation (\ref{eq:expr_for_fixed}) it follows, that $\chi_F(\Phi)$ is a commutant of $\mathcal{A}(\Phi)$. Von Neumann bicommutant theorem implies, that $\mathcal{A}$ is a commutant of $\chi_F(\Phi)$ and therefore is also a $\star$-algebra.   

$4 \Rightarrow 1$. If $\mathcal{A}(\Phi)$ is a $\star$-algebra, then using decomposition from Corollary \ref{wn:dekompozycja_unitalnych} one can define irreducible CP maps $\Phi^{\circ}_m(X) = \sum_{i=1}^K A^{\dagger}_{im} X A_{im}$, $i=1, \ldots N$. From Perron - Frobenius theorem, map $\Phi^{\circ}_m$ has strictly positive fixed point $\rho_m \in \mathcal{B}(\mathcal{H}_m)$. Now the matrix
\begin{equation}
\rho = 
\begin{bmatrix}
\rho_{1} & 0 & \hdots & 0 \\
0      &    \rho_{2} & \hdots & 0\\
\vdots & \hdots & \ddots & 0 \\
0 & \ldots & 0 & \rho_{N} \\
\end{bmatrix}
\end{equation}
is strictly positive fixed point of $\Phi^{\circ}$. 
\end{proof}
Theorem \ref{thm:summary} assumes unitality - however, in quantum information one usally operates on trace preserving CP - maps. Reference \cite{jex} provides the ,,dual'' version of Theorem \ref{thm:summary} (Thm. V.2). To emphasize the connections, we present it in the form analogous to Theorem \ref{thm:summary} (the proof uses similar techniques):

\begin{theorem}\label{thm:novotny}
Let $\Phi$ be trace preserving, completely positive map given by its Kraus decomposition $\Phi(X) =  \sum_{i=1}^K A_i X A_i^{\dagger}$. Then, the following conditions are equivalent:
\begin{enumerate}
\item The Kraus algebra $\mathcal{A}(\Phi)$ is a $\star$-algebra. 
\item The map $\Phi$ has full rank fixed point, that is there exists $\rho > 0$ such that $\Phi(\rho) = \rho$. Moreover, following equalities are true:
\begin{equation}
\chi(\Phi) = \lbrace X \in \mathcal{M}_n(\mathbb{C}) \colon \, \exists \lambda \in \mathbb{C} \, \forall i \colon \, A_i X \rho^{-1} - \lambda X \rho^{-1} A_i = 0, \, \vert \lambda \vert=1 \rbrace. 
\end{equation}
\begin{equation}
\chi_F(\Phi) = \lbrace X \in \mathcal{M}_n(\mathbb{C}) \colon \,  \forall i \colon \, X \rho^{-1} A_i - A_i X \rho^{-1} = 0 \rbrace.
\end{equation}
Thus, the spaces $\mathrm{span} \, \chi(\Phi)$ and $\chi_F(\Phi)$ are $\star$-algebras with respect to modified product $X \odot Y = X \rho^{-1} Y$. 
\end{enumerate}
\end{theorem}

Theorems \ref{thm:summary} and \ref{thm:novotny} exhibit connections between the structure of the sets $\mathcal{A}(\Phi)$, $\chi(\Phi)$ and $\chi_F(\Phi)$ and give neat expressions for them - however, they give no information about possible values of $\lambda$. We show how to exploit the structure of $\mathcal{A}(\Phi)$ to draw some conclusions from it. Basic result, used as the starting point for application of Shemesh and Amitsur - Levitzki theorems in section III is stated in Theorem \ref{wn:nierownewymiary}. The proof of this theorem uses following lemma (which proof, using the method from Ref. \cite{kribs} can be found in Appendix):

\begin{lem}\label{lem:dist_dim}
Let $\Phi(Y) = \sum_{i=1}^K A_i Y A_i^{\dagger}$ be a unital $CP$ map such that $\mathcal{A}(A_1,\ldots A_K)$ is a $^{\star}$-algebra. If $P_k$ is an orthogonal projector onto the subspace $\mathcal{H}_k$ from equation (\ref{eq:dekompozycja_hilberta}), $X$ is an eigenvector to unimodular eigenvalue and for some $i \neq j$ one has $\mathrm{dim} \mathcal{H}_i \neq \mathrm{dim} \mathcal{H}_j$, then
\begin{equation}
\tilde{X}_{ij} = P_i\, X \, P_j = 0.
\end{equation}
\end{lem}

When the dimensions of spaces in the block - diagonal decomposition (\ref{eq:dekompozycja_hilberta}) are pairwise distinct, the eigenequation corresponding to a unimodular eigenvalue decouples to the eigenequations for smaller, irreducible operators.

\begin{theorem}\label{wn:nierownewymiary}
Let $\Phi \, \colon \mathcal{M}_n(\mathbb{C}) \rightarrow \mathcal{M}_n(\mathbb{C})$ be a unital or trace preserving CP map given in Kraus form:
\begin{equation}
\Phi(X) = \sum_{i=1}^K A_i X A_i^{\dagger}.
\end{equation}
If an algebra $\mathcal{A}(A_1,\ldots, A_K)$ is $\star$-algebra and in its decomposition from Corollary \ref{wn:dekompozycja_unitalnych} all non-zero blocks have pairwise different dimensions (that is, $d_i\neq d_j$ for $i\neq j$, $i,j=1,\ldots, N$), then there exist natural numbers $m_i$, $1 \leq m_i \leq d_i^2$ such that
\begin{equation}
\mathrm{spec}_1(\Phi) = C_{m_1} \cup \ldots \cup C_{m_N},
\end{equation}
where $C_p$ denotes the cyclic subgroup of $U(1)$ of order $p$, i.e the set $\lbrace e^{\frac{2\pi i k}{p}}, \: k=0,\ldots p-1 \rbrace$. The asymptotic dynamics of the iterated quantum map is periodic with period at most $\mathrm{LCM}(m_1,\ldots,m_N)$,
 where $\mathrm{LCM}$ denotes the least common multiple. Moreover, if $\mathrm{span} \lbrace A_i \rbrace_{i=1}^K$ contains an invertible element, then all $m_i$ are divisors of $n$. 
\end{theorem}
\begin{proof}
Because the map and its dual have the same spectrum, we can wlog assume that $\Phi$ is unital. Let $X\in \mathcal{M}_n(\mathbb{C})$ and $|\lambda|=1$ satisfy the equation:
\begin{equation}
\Phi(X) = \lambda X.
\end{equation}
If $\tilde{X}_{jk}$ are defined as in Lemma \ref{lem:dist_dim}, then from the assumption $\tilde{X}_{jk} = 0$ for $j \neq k$. Consequently, if $X\neq 0$ then there exists $m \in \lbrace 1,\ldots, N \rbrace$ such that
\begin{equation}
\sum_{i=1}^K A_{im} \tilde{X}_{mm} A_{im}^{\dagger} = \lambda \tilde{X}_{mm}.
\end{equation}
Consider the map $\Phi_m\,\colon \mathcal{M}_{d_m}(\mathbb{C}) \rightarrow \mathcal{M}_{d_m}(\mathbb{C})$, 
\begin{equation} \label{eq:nonequalblocks}
\Phi_m(Y) = \sum_{i=1}^K A_{im} Y A_{im}^{\dagger}.
\end{equation}
We know that $\mathcal{A}(A_{1m},A_{2m},\ldots A_{Km}) = \mathcal{M}_{d_m}(\mathbb{C})$, therefore $\Phi_m$ is irreducible and from the Theorem \ref{thm:Groh} we have $\lambda \in C_p$, where $p \in \lbrace 1,\ldots, d_m^2 \rbrace$. If $\mathrm{span} \lbrace A_i \rbrace_{i=1}^K$ contains an invertible element, then the conclusion follows from Theorem \ref{thm:stronger}. 
\end{proof}

It is worth noting, that a similar approach was earlier used by 
Wolf and Perez--Garcia\cite{inverse_eigenvalue}, 
who investigated the structure of the space ${\rm span} \chi(\Phi)$
and showed that any trace preserving map $\Phi$ acts on $\chi(\Phi)$ by unitary conjugation and 
permutation within the blocks of the same dimension.
In this case the cyclic eigenvalues 
 correspond to  the cycles of the permutation.
%
%
On the other hand, formulation of Theorem \ref {wn:nierownewymiary} above 
is directly based on the set of Kraus operators $A_i$ 
and therefore it is possible to check effectively if its hypothesis is satisfied 
using only matrix multiplication and addition.
Therefore the method of analyzing the spectrum of superoperators
 presented here is easier to handle with for physical applications
 as demonstrated below
 in Corollaries \ref{qutrits}, \ref{wn:nierowne_wymiary5}, \ref{wn:nierowne_wymiary4} 
 and Examples \ref{ex:23} and \ref{ex:33}.%
%

\section{Application to quantum channels}

Now the question is how we can investigate the structure of the algebra $\mathcal{A}(\Phi)$, in order to check the dimensionality of the blocks in the decomposition of the algebra (we are dealing with maps for which $\mathcal{A}(\Phi)$ is a $^{\star}$-algebra, so such decomposition exists). We show that one can use several tools to ,,probe'' the internal structure of the algebra using only its generators, that is, its Kraus operators. The most important tools are the Shemesh criterion and the Amitsur - Levitzki theorem. Furthermore, we present the example demonstrating how to conveniently find a basis of an algebra and draw conclusions concerning its structure. We begin with the simple corollary for $n=3$ : 

\begin{corollary}\label{qutrits}
Let $\Phi$ be unital quantum channel $\mathcal{M}_3(\mathbb{C}) \rightarrow \mathcal{M}_3(\mathbb{C})$ (i.e. defined on qutrits) which has the Kraus operators 
$A_1,\ldots, A_K$. If there exist $i,j$ such that $[A_i,A_j] \neq 0$ then the peripheral spectrum is a cyclic subgroup of $U(1)$ of order at most 9. If, moreover, $\mathrm{span} \lbrace A_i \rbrace_{i=1}^K$ contains an invertible element, the order or peripheral spectrum is at most 3. 

\end{corollary}
\begin{proof}
Let us write all possible divisions of $\mathcal{A}(A_1,\ldots, A_K)$ into blocks from Corollary \ref{wn:dekompozycja_unitalnych}, namely $3=1+1+1, \, 3=1+2, \, 3=3$. The case $3=1+1+1$ is impossible, because $\mathcal{A}$ is non-commutative. Case $3=3$, means, that the map $\Phi$ is irreducible, and in the case $3=2+1$ blocks have pairwise distinct dimensions so we are able to use the Theorems \ref{wn:nierownewymiary} and \ref{thm:stronger}.   
\end{proof}

We now introduce the above mentioned theorems of Shemesh and Amitsur-Levitzki:

\begin{theorem}[Shemesh \cite{shemesh}]\label{th:shemesh}
Matrices $A,B \in \mathcal{M}_n(\mathbb{C})$ have common eigenvector if and only if 
\begin{equation} \label{eq:shemesh-ogolny}
\mathcal{M}=\bigcap_{k,l=1}^{n-1} \mathrm{ker}[A^k,B^l] \neq \lbrace 0 \rbrace.
\end{equation}
\end{theorem}

It is not difficult to show that $\mathcal{M}$ is the smallest subspace of $\mathcal{H} = \mathbb{C}^n$ which contains all common eigenvectors of the matrices 
$A$ and $B$.

At the same time the subspace $\mathcal{M}$ defined in (\ref{eq:shemesh-ogolny}) is the common invariant subspace with respect to $A$ and $B$ 
and on which $A$ and $B$ commute.
There is a generalization of this theorem for an arbitrary number of matrices \cite{pastuszak}:
\begin{theorem} \label{th:generalized_shemesh}
Assume that $H, A_1, \ldots A_s \in \mathcal{M}_n(\mathbb{C})$ and that $H$ has pairwise distinct eigenvalues. Define:
\begin{equation}
\mathcal{N}(H,A_1,\ldots A_s) = \bigcap_{k=1}^{n-1} \bigcap_{i=1}^s \mathrm{ker}(H^k,A_i).
\end{equation}
Then, the matrices $H, A_1, \ldots A_k$ have a common eigenvector if and only if $\mathcal{N}(H,A_1,\ldots A_s) \neq \lbrace 0 \rbrace$. 
\end{theorem}
The \emph{standard polynomial} for $n$ noncommutative variables $X_1,\ldots X_n$ is defined in the following way: 
\begin{equation}
S_n(X_1,\ldots , X_n) = \sum_{\sigma \in \mathcal{S}_n} \mathrm{sign}(\sigma) X_{\sigma(1)}\ldots X_{\sigma(n)},
\end{equation}
where the summation runs over all permutations of the set  $\lbrace 1,\dots ,n \rbrace$.

\begin{theorem}[Amitsur, Levitzki \cite{amitsur1}]\label{thm:amitsur-levitzki}
The full matrix algebra $\macierze$ satisfies the standard polynomial identity of order $2n$, that is, for all matrices $A_1,\ldots A_{2n}$ we have :
\begin{equation}
S_{2n}(A_1,\ldots A_{2n}) = 0.
\end{equation}
Moreover the algebra $\macierze$ satisfies no identity of order smaller than $2n$.  
\end{theorem}

The following corollaries exploit similar idea to the one for $n=3$:

\begin{corollary}\label{wn:nierowne_wymiary5}
Let $\Phi$ be a unital quantum channel $\mathcal{M}_5(\mathbb{C}) \rightarrow \mathcal{M}_5(\mathbb{C})$ with Kraus operators $A_1,A_2$. If:
\begin{equation}\label{eq:warunek_shemesh}
\bigcap_{k,l=1}^4 \mathrm{ker}[A_1^k,A_2^l]=\lbrace 0 \rbrace,
\end{equation}
then the peripheral spectrum of $\Phi$ is the sum of two cyclic subgroups of $U(1)$ of order at most 25. If, moreover, $\mathrm{span} \lbrace A_i \rbrace_{i=1}^K$ contains an invertible element, then the order of peripheral spectrum is at most 5.
\end{corollary}
\begin{proof}
Condition (\ref{eq:warunek_shemesh}) means, that $A_1,A_2$ do not have common eigenvector. The only divisions into blocks can assume the form $5=5$ or $5=2+3$ and it remains to use Theorems \ref{wn:nierownewymiary} and \ref{thm:stronger}. 
\end{proof}

\begin{corollary}\label{wn:nierowne_wymiary4}
Let  $\Phi$ be the mixed unitary channel $\mathcal{M}_4(\mathbb{C}) \rightarrow \mathcal{M}_4(\mathbb{C})$ given in the form:
\begin{equation}
\Phi(X) = \sum_{i=1}^K p_iU_i X U_i^{\dagger}.
\end{equation} 
where $U_i$ are unitary, $p_1,\ldots p_K \geq 0$ and $\sum_{i=1}^K p_i = 1$. Furthermore, define $V_{ij}$ as follows:
\begin{equation}
V_{ij} = U_i U_j U_i^{\dagger} U_j^{\dagger} - 
U_i^{\dagger}U_j U_i U_j^{\dagger} +
U_i^{\dagger}U_j^{\dagger} U_i U_j-
U_i U_j^{\dagger} U_i^{\dagger} U_j.
\end{equation}
If there exist $i,j \in \lbrace 1,\ldots, K \rbrace$ such that
\begin{equation}\label{eq:warunek_z_amitsura}
V_{ij}+V_{ij}^{\dagger} \neq 0
\end{equation} 
then the peripheral spectrum of $\Phi$ is the cyclic subgroup of $U(1)$ of order at most 4. 
\end{corollary}

\begin{proof}
We will show that the condition from above assumption means that $\mathcal{A}(U_1, \ldots U_K)$ does not satisfy the standard polynomial identity $S_4$. The map $\Phi$ is unital, therefore $\mathcal{A}(U_1,\ldots U_K)$ is a $^{\star}$-algebra and contains $U_i^{\dagger}$ for $i=1,\ldots, K$. For $i\neq j$ we will compute $S_4(U_i,U_j,U_i^{\dagger},U_j^{\dagger})$. 
To this end, rewrite $U_1\rightarrow 1$, $U_2\rightarrow 2$, $U_1^{\dagger} \rightarrow 3$, $U_2^{\dagger}\rightarrow 4$. Observe, that from unitarity, all permutations, where numbers $1,3$ or $2,4$ are side by side give the expression $\mathbb{I}_4$ in $S_4$. Moreover, for each ,,one'' with minus sign there exists a ,,one'' with plus sign (e.g. $(1324)$ and $(3124)$). All ,,ones'' are created from the permutations $(1324)$, $(2413)$, $(1243)$, $(2134)$ due to the pairwise exchanges $1 \leftrightarrow 3$, $2 \leftrightarrow  4$ - together we have  $16$ permutations. We are left with $8$ permutations: $(1234)$, $(3214)$, $(3412)$, $(1432)$ and their conjugations. Substituting matrices for numbers we obtain:
\begin{equation}
S_4(U_i,U_j,U_i^{\dagger},U_j^{\dagger}) = V_{ij}+V_{ij}^{\dagger} \neq 0
\end{equation}
From Theorem \ref{thm:amitsur-levitzki} we can conclude, that the algebra $\mathcal{A}(U_1,\ldots U_K)$ in its decomposition can not contain the algebra 
$\mathcal{M}_2(\mathbb{C})$. Therefore, the only possible decompositions are the following: $4=4, \,4=3+1$ - it remains to use Theorems \ref{wn:nierownewymiary} and \ref{thm:stronger}. 
\end{proof}

\begin{remark} \label{rem:generalization}
\emph{The assumptions in Corollaries \ref{qutrits}, \ref{wn:nierowne_wymiary5} and \ref{wn:nierowne_wymiary4} can be weakened. It suffices that the map $\Phi$ is completely positive unital or trace preserving, for which the algebra generated by its Kraus operators is a $^{\star}$-algebra (then, the hypothesis of Theorem \ref{wn:nierownewymiary} is still satisfied)}. 
\end{remark} 

\begin{remark}
\emph{The methods presented in proofs of corollaries \ref{qutrits} are flexible and applicable to other dimensions (one can use standard polynomials of higher order). Moreover, for  a trace preserving and unital channel $\Phi$ one can effectively check its irreducibility: it suffices to verify if the fixed point spaces $\chi_F(\Phi)$ and $\chi_F(\Phi^{\circ})$ are one-dimensional\cite{lecture_notes} - then $\mathbb{I}$ is the only strictly positive fixed point of $\Phi$ and $\Phi^{\circ}$. For example, this criterion can help to decide which decomposition ($2+3$ or $5$) is realized in the setup of Corollary \ref{def:dual} when (\ref{eq:warunek_shemesh}) is satisfied: if $\mathrm{dim}\, \chi_F(\Phi) > 1$ or $\mathrm{dim}\, \chi_F(\Phi^{\circ}) > 1$, then one has decomposition $2+3$ and the peripheral spectrum is a sum of two cyclic subgroups of $U(1)$ of order at most 9 (at most 3 if linear span of Kraus operators contains an invertible element).} 
\end{remark} 
The above statement emphasizes the important role of the \emph{algebraic structure} in the theorems presented here. 

%

\section{Examples}

\subsection{Application of the Shemesh criterion}

\begin{ex}
\label{ex:23}
Let us consider quantum channel given by the following Kraus operators:
\begin{equation}
A_1 = 
\left[
\begin{array}{ccc}
 \frac{3}{10} & \frac{i \sin \phi }{\sqrt{2}} & -\frac{3}{10} \\
 -\frac{i \sin \phi }{\sqrt{2}} & 0 & -\frac{i \sin \phi }{\sqrt{2}} \\
 -\frac{3}{10} & \frac{i \sin \phi }{\sqrt{2}} & \frac{3}{10} \\
\end{array}
\right],
\end{equation}
\begin{equation}
A_2 = \left[
\begin{array}{ccc}
 \frac{2}{5}-\frac{\cos \phi }{2} & 0 & -\frac{\cos \phi }{2}-\frac{2}{5} \\
 0 & \cos \phi  & 0 \\
 -\frac{\cos \phi }{2}-\frac{2}{5} & 0 & \frac{2}{5}-\frac{\cos \phi }{2} \\
\end{array}
\right],
\end{equation}
where the parameter $\phi \in [0, 2\pi)$. 
One can check, that 
\begin{equation}
A_1^{\dagger} \, A_1 + A_2^{\dagger} \, A_2 = A_1 \, A_1^{\dagger} + A_2 \, A_2^{\dagger} = \mathbb{I},
\end{equation}
so the map is trace preserving and unital and consequently $\mathcal{A}(A_1, A_2)$ is a $^{\star}$-algebra. Moreover, one has : 
\begin{equation}
[A_1, A_2] =
\left[
\begin{array}{ccc}
 0 & i \sqrt{2} \sin \phi  \cos \phi  & 0 \\
 i \sqrt{2} \sin \phi  \cos \phi  & 0 & i \sqrt{2} \sin \phi  \cos \phi  \\
 0 & i \sqrt{2} \sin \phi  \cos \phi  & 0 \\
\end{array}
\right]
\end{equation} 
and
\begin{equation}
[A_1^2,A_2]  = [A_1,A_2^2] = 0,
\end{equation} 
so $\mathrm{dim}$ $\mathrm{ker}[A_1, A_2]=1$ and it follows from Shemesh theorem that matrices $A_1$, $A_2$ have one common eigenvector. Therefore the algebra $\mathcal{A}(A_1,A_2)$ decomposes into blocks with dimensions $2$ and $1$. Moreover, for $\phi \notin \lbrace 0, \pi \rbrace$ matrices $A_1, A_2$ are invertible, so only $1$ and $-1$ can belong to the peripheral spectrum. In fact, we can conclude, that \emph{for every $\phi$ the peripheral spectrum of $\Phi$ is exactly equal to the set $\lbrace 1, -1 \rbrace$}. To see this, one can observe that $2 \times 2$ part of the algebra $\mathcal{A}(A_1, A_2)$ is a unital qubit map, therefore it can be expressed using Pauli matrices as its Kraus operators (including identity matrix). However, one can check that for every choice of two matrices from the set $\lbrace \mathbb{I}, \sigma_x, \sigma_y, \sigma_z \rbrace$ the assumptions of Theorem \ref{thm:primitive} are not satisfied - therefore the map $\Phi$ is not primitive and its peripheral spectrum must be nontrivial.       

\end{ex}

\begin{ex}
\label{ex:33}
Let us now consider unital quantum map given by three Kraus operators:
\begin{equation}
\nonumber
A_1 = \left[
\begin{array}{ccc}
 \frac{1}{4 \sqrt{2}} & -\frac{1}{4 \sqrt{2}}-\frac{i \sin \phi }{\sqrt{2}} & \frac{1}{4}-\frac{1}{2} i \sin \phi  \\
 -\frac{1}{4 \sqrt{2}}+\frac{i \sin \phi }{\sqrt{2}} & \frac{1}{4 \sqrt{2}} & -\frac{1}{4}-\frac{1}{2} i \sin \phi  \\
 \frac{1}{4}+\frac{1}{2} i \sin \phi  & -\frac{1}{4}+\frac{1}{2} i \sin \phi  & \frac{1}{2 \sqrt{2}} \\
\end{array}
\right],
\end{equation}
\begin{equation}
\nonumber
A_2 = \left[
\begin{array}{ccc}
 \frac{1}{4 \sqrt{2}}-\frac{\cos \phi }{4} & -\frac{3 \cos \phi }{4}-\frac{1}{4 \sqrt{2}} & \frac{1}{4}-\frac{\cos \phi }{2 \sqrt{2}} \\
 -\frac{3 \cos \phi }{4}-\frac{1}{4 \sqrt{2}} & \frac{1}{4 \sqrt{2}}-\frac{\cos \phi }{4} & \frac{\cos \phi }{2 \sqrt{2}}-\frac{1}{4} \\
 \frac{1}{4}-\frac{\cos \phi }{2 \sqrt{2}} & \frac{\cos \phi }{2 \sqrt{2}}-\frac{1}{4} & \frac{\cos \phi }{2}+\frac{1}{2 \sqrt{2}} \\
\end{array}
\right],
\end{equation}
\begin{equation}
\nonumber
A_3 = \left[
\begin{array}{ccc}
 \frac{1}{4 \sqrt{2}} & -\frac{1}{4 \sqrt{2}}+\frac{1}{2} i \sin \phi  & \frac{1}{4}+\frac{i \sin \phi }{2 \sqrt{2}} \\
 -\frac{1}{4 \sqrt{2}}-\frac{1}{2} i \sin \phi  & \frac{1}{4 \sqrt{2}} & -\frac{1}{4}+\frac{i \sin \phi }{2 \sqrt{2}} \\
 \frac{1}{4}-\frac{i \sin \phi }{2 \sqrt{2}} & -\frac{1}{4}-\frac{i \sin \phi }{2 \sqrt{2}} & \frac{1}{2 \sqrt{2}} \\
\end{array}
\right].
\end{equation}
Note that in this case one can not use the basic version (\ref{eq:shemesh-ogolny}) of the Shemesh theorem. However, we can make use of Theorem \ref{th:generalized_shemesh} - this theorem gives a computable criterion to check whether matrices have a common eigenvector, one can verify if matrix has pairwise distinct eigenvalues by computing the discriminant. For example, in our case the discriminant of the characteristic polynomial of $A_3$ is equal to
\begin{equation}
\mathrm{disc}(A_3) = \frac{1}{2} \left(\sin ^6 \phi - 2 \sin ^4 \phi + \sin^2 \phi \right),
\end{equation}
and is nonzero in general. Thus, $A_3$ has pairwise distinct eigenvalues, and therefore we can apply Theorem \ref{th:generalized_shemesh} (the role of $H$ is 
played by $A_3$). We have that:
\begin{equation}
[A_3^2,A_1] = [A_3^2,A_2] = 0
\end{equation}
and
\begin{equation}
\mathrm{ker} [A_3,A_1] = \mathrm{ker}[A_3,A_2] = \left(\frac{1}{\sqrt{2}}, \frac{-1}{\sqrt{2}},1 \right),
\end{equation}
Therefore, the matrices $A_1$, $A_2$ and $A_3$ have a common eigenvector, the algebra $\mathcal{A}(A_1,A_2,A_3)$ has block decomposition in the form $2+1$ and we know that 
the peripheral spectrum is the cyclic subgroup of $U(1)$ of order at most 4. Now, applying similar reasoning as in the discussion of previous example we can conclude, that \emph{for every $\phi   \notin \lbrace 0, \pi \rbrace$ the peripheral spectrum of $\Phi$ is equal to $\lbrace 1 \rbrace$}. Indeed, for every choice of three matrices from the set $\lbrace \mathbb{I}, \sigma_x, \sigma_y, \sigma_z \rbrace$ assumptions of Theorem \ref{thm:primitive} are satisfied for $m=2$, so the $2 \times 2$ part of the algebra $\mathcal{A}(A_1, A_2, A_3)$ represents the primitive map.    
\end{ex}

\subsection{Finding the basis of the algebra}

The aim of this section is to show how to use Remark \ref{rem:generalization} in practice. Applying this technique one can get information concerning the peripheral spectrum without assuming unitality, however, first one has to establish whether algebra $\mathcal{A}(\Phi)$ is a $^{\star}$-algebra.

We argue that in some cases it is quite easy to find the basis and check the properties of the algebra $\mathcal{A}$ explicitly. For example, if one has an
algebra of $n \times n$ matrices whose basis consists of $D < n^2$ elements, then one knows that this algebra is reducible. 

Suppose we want to find the basis of the algebra $\mathcal{A}(A_1, A_2)$, where $A_1$ and $A_2$ are square $n \times n$ matrices. One has to consider the following sequence of words formed by $A_1$ and $A_2$:
\begin{equation}
\mathbb{I}, \,  A_1, \,  A_2, \,  A_1^2, \,  A_2^2, \, A_1\, A_2, \, A_2 \, A_1, \ldots.
\end{equation}
We can think about each element of this sequence as a word $w_m(A_1, A_2)$ of length $m$, that is, as:
\begin{equation}
\label{eq:def_slowa}
w_i(A_1, A_2) = \lbrace A_{j_1}^{i_1}A_{j_2}^{i_2} \ldots A_{j_k}^{i_k} , \, i_1+i_2+\ldots i_k = i \rbrace
\end{equation}

Denoting the set of all words of length not greater than $m$ as $L_m$ and applying the Cayley - Hamilton theorem, we obtain for some $p$, that
\begin{equation}
L_0 \subset L_1 \ldots \subset L_p = L_{p+1} = \mathcal{A}(A_1, A_2).
\end{equation}
Thus $p \leq n^2-1$. However, one can find much better upper bound \cite{paz, pappacena}, for example:
\begin{equation}\label{eq:papacena}
p \leq \left \lceil \frac{n^2 + 3}{2} \right \rceil,
\end{equation}
which holds for all generating matrices. In the next example we present how one can investigate the structure of the generated algebra systematically (using basic functionalities of tools like \emph{Mathematica} or \emph{Matlab}). 

Let us consider the quantum channel given by the following Kraus operators of size $n=3$:
\begin{equation}
A_1 = 
\frac{1}{\sqrt{6}} \left[
\begin{array}{ccc}
 1 & -1 & 1 \\
 \frac{1}{\sqrt{2}} & \sqrt{2} & 0 \\
 -\frac{1}{\sqrt{2}} & 0 & \sqrt{2} \\
\end{array}
\right],
\end{equation}
\begin{equation}
A_2 = 
\frac{1}{\sqrt{6}} \left[
\begin{array}{ccc}
 \sqrt{2} & \frac{1}{\sqrt{2}} & -\frac{1}{\sqrt{2}} \\
 -1 & \frac{3}{2} & \frac{1}{2} \\
 1 & \frac{1}{2} & \frac{3}{2} \\
\end{array}
\right].
\end{equation}

One can easily check, that $A_1^{\dagger}\, A_1 + A_2^{\dagger}\, A_2 = \mathbb{I}$, so, indeed, those operators represent a quantum channel. However, as $A_1 \, A_1{\dagger} + A_2 \, A_2^{\dagger} \neq \mathbb{I}$, this channel is not \emph{unital} and corollary \ref{qutrits} can not be applied directly. It is also not clear whether they generate a $^{\star}$ - algebra. However, using inequality \ref{eq:papacena}, one can observe that it suffices to consider the ,,words'' up to order $4$. We follow the procedure described in the general instructions and proceed as follows: 
\begin{enumerate}
\item Compute the products $w_i(A_1,A_2)$ up to order $i = 4$. 
\item Construct a rectangular matrix, later called $R$ with columns obtained by reshaping square matrices $w_i(A_1,A_2)$ computed in item 1.
\item Perform the Gaussian elimination on $R$ in order to find the basis.
\end{enumerate}
In the considered example the basis turns out to consist of the matrices $A_1, \, A_2, \, A_1^2, A_1\,A_2, A_2\,A_1$:
\begin{equation*}
E_1 = A_1 = \left[
\begin{array}{ccc}
 1 & -1 & 1 \\
 \frac{1}{\sqrt{2}} & \sqrt{2} & 0 \\
 -\frac{1}{\sqrt{2}} & 0 & \sqrt{2} \\
\end{array}
\right], \; E_2 = A_2 = \left[
\begin{array}{ccc}
 \sqrt{2} & \frac{1}{\sqrt{2}} & -\frac{1}{\sqrt{2}} \\
 -1 & \frac{3}{2} & \frac{1}{2} \\
 1 & \frac{1}{2} & \frac{3}{2} \\
\end{array}
\right],
\end{equation*}
\begin{equation*}
E_3 = A_1^2 = \left[
\begin{array}{ccc}
 1-\sqrt{2} & -1-\sqrt{2} & 1+\sqrt{2} \\
 1+\frac{1}{\sqrt{2}} & 2-\frac{1}{\sqrt{2}} & \frac{1}{\sqrt{2}} \\
 -1-\frac{1}{\sqrt{2}} & \frac{1}{\sqrt{2}} & 2-\frac{1}{\sqrt{2}} \\
\end{array}
\right], \;  E_4 = A_1 \, A_2 = \left[
\begin{array}{ccc}
 2+\sqrt{2} & -1+\frac{1}{\sqrt{2}} & 1-\frac{1}{\sqrt{2}} \\
 1-\sqrt{2} & \frac{1}{2}+\frac{3}{\sqrt{2}} & -\frac{1}{2}+\frac{1}{\
\sqrt{2}} \\
 -1+\sqrt{2} & -\frac{1}{2}+\frac{1}{\sqrt{2}} & \
\frac{1}{2}+\frac{3}{\sqrt{2}} \\
\end{array}
\right]
\end{equation*}
\begin{equation*}
E_5 = A_2 \, A_1 = \left[
\begin{array}{ccc}
 1+\sqrt{2} & 1-\sqrt{2} & -1+\sqrt{2} \\
 -1+\frac{1}{\sqrt{2}} & 1+\frac{3}{\sqrt{2}} & -1+\frac{1}{\sqrt{2}} \
\\
 1-\frac{1}{\sqrt{2}} & -1+\frac{1}{\sqrt{2}} & 1+\frac{3}{\sqrt{2}} \\
\end{array}
\right].
\end{equation*}
Therefore:
\begin{equation}
D = \mathrm{dim}(\mathcal{A}(A_1,A_2)) < n^2,
\end{equation}
entailing that the algebra is reducible - but without any information whether it is a $^{\star}$-algebra or even a semisimple algebra. 
Performing the Gaussian elimination we can observe, that:
\begin{equation}
A_1^{\dagger} = \frac{1}{\sqrt{6}} \left(a_1\, E_1 + a_2\, E_2 + a_4\, E_4 + a_5 \, E_5\right) = a_1 A_1 + a_2 A_2 + a_4 A_1 A_2 + a_5 A_2 A_1 
\end{equation}
where:
\begin{equation*}
a_1 = \frac{1}{14} \left(3-8 \sqrt{2}\right), \quad a_2 = \frac{1}{28} \left(6+5 \sqrt{2}\right), \quad a_4 = \frac{1}{2 \sqrt{2}-8}, \quad a_5 = \frac{1}{14} \left(5+3 \sqrt{2}\right).
\end{equation*}
Analogously, 
\begin{equation}
A_2^{\dagger} = \frac{1}{\sqrt{6}} \left(b_1\, E_1 + b_2\, E_2 + b_4\, E_4 + b_5 \, E_5 \right) = b_1 A_1 + b_2 A_2 + b_4 A_1 A_2 + b_5 A_2 A_1 ,  
\end{equation}
where
\begin{equation*}
b_1 = \frac{1}{14} \left(11+8 \sqrt{2}\right), \quad b_2 = \frac{1}{28} \left(22-5 \sqrt{2}\right), \quad b_4 = \frac{1}{28} \left(4+\sqrt{2}\right), \quad b_5 = \frac{1}{14} \left(-5-3 \sqrt{2}\right)
\end{equation*}
Given the above equations it seems quite obvious that the algebra $\mathcal{A}(A_1,A_2)$ is a $^{\star}$-algebra. Moreover, from the trace  preservation conditions, we know it contains an
identity matrix. It is now clear that the only possible division into irreducible blocks is $3 = 2+1$ (to generate $5$ - dimensional algebra), meeting the requirements of Remark 1.  

\section{Concluding remarks}
The main objective of this paper was to present some examples and practical tools enabling us 
to ,,probe'' the structure of the algebra generated by the given matrices which determine quantum channels. 

In section 2 we discussed how the classical theorem of Groh \cite{groh}, concerning the peripheral spectrum of irreducible channel, can be generalized to cover channels that are not necessarily irreducible. We showed that by assuming that the algebra generated by the Kraus operators is a $^{\star}$-algebra and that its block decomposition consists of blocks of pairwise distinct dimensions, the peripheral spectrum contains only roots of unity. In section 3 we provided operational criteria based on the Shemesh and Amitsur - Levitzki theorems which allow one to analyze the structure of the 
algebra and to verify validity of the assumptions concerning different dimensionalities of the blocks.  

In section 4 we demonstrated a practical example of dealing with the generalized case of \emph{non unital} quantum
maps by applying  Remark \ref{rem:generalization}. 
The techniques  presented in this work allow us to gain an insight into the spectrum of the 
superoperator corresponding to the map by using only elementary arithmetic operations, which include 
multiplication and addition of matrices and solving linear equations. 
Therefore, one can apply these methods, for instance, 
to analyze entire families of superoperators dependent on given parameters.

The physical importance of the presented topic is closely related to 
iterated quantum dynamics (multiple application of a given quantum operation).
Knowing eigenvalues of the superoperator belonging to the
peripheral spectrum one can predict the maximal length 
of the cycles which may occur during the asymptotic dynamics.
Described methods of investigating the structure of matrix algebra 
could also be applied to the problem of quantum compression \cite{bluhm}. 

 
 Let us conclude the paper by presenting a short list of open questions.
 The subject of characterizing the peripheral spectrum of an irreducible but not primitive quantum channel
 is left over for a further study.
  It could be also interesting to examine the connection between the number of eigenvalues in 
 the peripheral spectrum and the maximal dimension of the space $S_m(\Phi)$ defined 
 in the hypothesis of the Theorem \ref{thm:primitive}. Furthermore, developing more effective algorithms 
 for computing the basis of the generated algebra and verifying its closeness under the Hermitian 
 conjugation would allow us to provide a detailed description of spectral properties of 
 quantum channels and to predict asymptotic properties of discrete quantum dynamics.

\section{Acknowledgments}

M.B. is grateful to Agnieszka Proszewska for insightful discussions and critical reading of the manuscript. M.B. also acknowledges the support of Adrian Stencel for remarks and comments. Financial support by the Polish National Science Center (NCN) under the grants number 2016/20/W/ST4/00314 (M.B.), 2015/19/B/ST1/03095 (A.J.) and DEC-2015/18/A/ST2/00274 (K.Ż) is also acknowledged.

\appendix
\section{Proof of Lemma \ref{lem:dist_dim}}

First we recall a simple fact concerning the matrix norms : Let $||\cdot||_p$ denote the standard norm $\ell_p$ in the space $\mathbb{C}^n$. Matrix $A\in \mathcal{M}_{m,n}(\mathbb{C})$ with $m$ rows and $n$ columns forms an operator $\mathbb{C}^n\rightarrow \mathbb{C}^m$. In the matrix space $\mathcal{M}_{m,n}(\mathbb{C})$ we have the standard norm induced by the norm $||\cdot||_p$:
\begin{equation}
\label{eq:matrix_norm_def}
||A||_p = \mathrm{sup}_{x\neq0}\frac{||Ax||_p}{||x||_p} = \mathrm{inf} \lbrace C\geq0 \, \colon \forall_{x\in \mathbb{C}^n} ||Ax||_p\leq C||x||_p \rbrace.
\end{equation}
Furthermore, the infinity norm $||\cdot||_{\infty}$ which is equal to the largest singluar value of the matrix. 

\begin{prop}\label{prop:norma_operatorowa}
If the vector 2-norm $||\cdot||_2$ is a norm in $\mathbb{C}^n$ induced by the Euclidean scalar product, then for each $A \in \mathcal{M}_{m,n}(\mathbb{C})$ one has :
\begin{equation}
||A||_2 = ||A||_{\infty}.
\end{equation}
where $||A||_2$ is the norm in matrix space defined by (\ref{eq:matrix_norm_def}). Moreover, the set
\begin{equation}
\mathcal{V} = \lbrace x \in \mathbb{C}^n \, \colon ||Ax||_2 = ||A||_2 \, ||x||_2 \rbrace
\end{equation}
is a vector subspace of $\mathbb{C}^n$. 
\end{prop}

\emph{Proof of Lemma \ref{lem:dist_dim}}:
Let assume that $|\lambda|=1$, $X\in \mathcal{M}_n(\mathbb{C}), X\neq 0$ and
\begin{equation}\label{eq:wlasne_peripheral}
\sum_{i=1}^K A_i X A_i^{\dagger} = \lambda X.
\end{equation} 
Using corollary  \ref{wn:dekompozycja_unitalnych} we choose the orthonormal basis  $\lbrace e_j \rbrace_{j=1}^n$ of the space $\mathcal{H}\simeq \mathbb{C}^n$ in which all $A_i$ are of the form (\ref{eq:dekompozycja_unitalnych}). Now it is clear that it suffices to show the thesis for the Kraus operators in this basis. Zero block is not present due to the unitality. Let $P_m$ be the hermitean projector onto the space $\mathcal{H}_m$ from the decomposition (\ref{eq:dekompozycja_hilberta}), $m=1,\ldots,N$. Denote $X_{jk} = P_j X P_k$ for $j,k=1,\ldots,n$. Equation (\ref{eq:wlasne_peripheral}) can be rewritten in the block form
\begin{equation}\label{eq:peripheral_blokowa}
\sum_{i=1}^K
\begin{bmatrix}
A_{i1} & 0 & \hdots & 0 \\
0      &    A_{i2} & \hdots & 0\\
\vdots & \hdots & \ddots & \vdots \\
0 & \hdots & 0 & A_{iN} \\
\end{bmatrix}
\begin{bmatrix}
\tilde{X}_{11} & \tilde{X}_{12}& \hdots & \tilde{X}_{1N} \\
\tilde{X}_{21} & \tilde{X}_{22} & \hdots & \tilde{X}_{2N} \\
\vdots & \hdots & \ddots & \vdots \\
\tilde{X}_{N1} & \hdots & \hdots &\tilde{X}_{NN}\\
\end{bmatrix}
\begin{bmatrix}
A_{i1}^{\dagger} & 0 & \hdots & 0 \\
0      &    A_{i2}^{\dagger} & \hdots & 0\\
\vdots & \hdots & \ddots & \vdots \\
0 & \hdots & 0 & A_{iN}^{\dagger} \\
\end{bmatrix}
=
\end{equation}
\begin{equation}
=\lambda
\begin{bmatrix}
\tilde{X}_{11} & \tilde{X}_{12}& \hdots & \tilde{X}_{1N} \\
\tilde{X}_{21} & \tilde{X}_{22} & \hdots & \tilde{X}_{2N} \\
\vdots & \hdots & \ddots & \vdots \\
\tilde{X}_{N1} & \hdots & \hdots & \tilde{X}_{NN}\\
\end{bmatrix}.
\end{equation}
We obtain $N^2$ equations:
\begin{equation}\label{eq:peripheral1}
\sum_{i=1}^K A_{ij} \tilde{X}_{jk} A_{ik} = \lambda \tilde{X}_{jk},\quad j,k=1,\ldots, N.
\end{equation}
Now we introduce matrices in the block form :
\begin{equation}
B_m = 
\begin{bmatrix}
A_{1m} & A_{2m} & \hdots & A_{Km}
\end{bmatrix}
,\quad
B_m \in \mathcal{M}_{d_m,Kd_m}(\mathbb{C}), \quad m=1,\ldots N
\end{equation}
and
\begin{equation}
\tilde{X}^{(K)}_{jk} = \mathrm{diag}(\underbrace{\tilde{X}_{jk},\tilde{X}_{jk},\ldots, \tilde{X}_{jk}}_{K}), \quad 
\tilde{X}^{(K)}_{jk} \in \mathcal{M}_{Kd_j,Kd_k}(\mathbb{C}).
\end{equation}
Matrices $B_m$ can be treated as the maps $\bigoplus_{i=1}^K\mathcal{H}_m \rightarrow \mathcal{H}_m$, and $\tilde{X}^{(K)}_{jk}$ as the maps  $\bigoplus_{i=1}^K\mathcal{H}_k \rightarrow \bigoplus_{i=1}^K\mathcal{H}_j$. 
 From unitality of $\Phi$ we see that the singular values of  $B_m$ equal 1 and from Proposition \ref{prop:norma_operatorowa} one has 
$||B_m||_2=||B_m^{\dagger}||_2=1$. Equations (\ref{eq:peripheral1}) can be rewritten in the form:
\begin{equation*}\label{eq:peripheral2}
B_j \tilde{X}^{(K)}_{jk} B_k^{\dagger} = \lambda \tilde{X}_{jk},\quad j,k = 1,\ldots, N.
\end{equation*}
Choose $j,k$ and assume, that $\tilde{X}_{jk} \neq 0$. Then $||\tilde{X}^{(K)}_{jk}||_2 = ||\tilde{X}_{jk}||_2 \neq 0$ as well. 
Using Proposition \ref{prop:norma_operatorowa} consider the subspace where $\tilde{X}_{jk}$ achieves its norm:
\begin{equation}
\mathcal{V}_{jk} = \lbrace v\in \mathcal{H}_{k}\, \colon ||\tilde{X}_{jk}\,v||_2 = ||\tilde{X}_{jk}||_2 \, ||v||_2 \rbrace \neq \lbrace 0 \rbrace.
\end{equation}
Take $v\in \mathcal{V}_{jk}$. Using the fact, that $||B_j||_2=||B_k^{\dagger}||_2=1$ and $|\lambda|=1$ we have 
\begin{equation*}
||\tilde{X}_{jk}||_2 \, ||v||_2 = ||\tilde{X}_{jk}\,v||_2=||\lambda \tilde{X}_{jk}v||_2 = 
||B_j \tilde{X}^{(K)}_{jk} B_k^{\dagger}v||_2 \leq ||\tilde{X}^{(K)}_{jk} B_k^{\dagger}v||_2 \leq^{(\star)}
\end{equation*}
\begin{equation}\label{eq:osiaganormerownosc}
\leq ||\tilde{X}^{(K)}_{jk}||_2\, ||B_k^{\dagger}v||_2 \leq 
||\tilde{X}_{jk}||_2\, ||v||_2.
\end{equation}
This means that inequality $(\star)$ above is saturated. But one has also
\begin{equation}
\tilde{X}^{(K)}_{jk} B_k^{\dagger}v = 
\begin{bmatrix}
\tilde{X}_{jk} A_{1k}^{\dagger} v \\
\tilde{X}_{jk} A_{2k}^{\dagger} v \\
\vdots \\
\tilde{X}_{jk} A_{Kk}^{\dagger} v \\
\end{bmatrix}, \quad
B_k^{\dagger} v = 
\begin{bmatrix}
A_{1k}^{\dagger} v \\
A_{2k}^{\dagger} v \\
\vdots \\
A_{Kk}^{\dagger} v\\
\end{bmatrix}.
\end{equation}
This implies
\begin{equation}
||\tilde{X}_{jk}A_{ik}^{\dagger} v||_2 =  ||\tilde{X}_{jk}||_2\,||A_{ik}^{\dagger} v||_2,\quad i=1,\ldots, K,
\end{equation}
and thus
\begin{equation}
A_{ik}^{\dagger} v \in  \mathcal{V}_{jk},\quad i=1,\ldots ,K.
\end{equation}
In consequence, $\mathcal{V}_{ik}$ is the common invariant subspace of $A_{ik}^{\dagger}$ for $i=1,\ldots,K$.       
 But $\mathcal{A}(A_{1k}^{\dagger},\ldots, A_{Kk}^{\dagger}) = \mathcal{A}(A_{1k},\ldots, A_{Kk}) = \mathcal{M}_{d_k}(\mathbb{C})$, therefore due to Burnside theorem we have $\mathcal{V}_{jk} = \mathcal{H}_{k} \simeq \mathbb{C}^{d_k}$ (since $\mathcal{V}_{jk}$ is nontrivial). Hence for $v\in \mathcal{H}_k$ one has $||\tilde{X}_{jk} v||_2 = ||\tilde{X}_{jk}||_2\,||v||_2$, and consequently
\begin{center}
\emph{$\frac{\tilde{X}_{jk}}{||\tilde{X}_{jk}||_2}$  is an isometry  $\mathcal{H}_k\rightarrow \mathcal{H}_j$                    $(\star\star)$}.
\end{center}
Repeating the above steps for $\tilde{X}_{jk}^{\dagger}$, which can be treated as an operator $\mathcal{H}_j\rightarrow \mathcal{H}_k$ we have 
\begin{center}
\emph{$\frac{\tilde{X}_{jk}^{\dagger}}{||\tilde{X}_{jk}^{\dagger}||_2}$  is an isometry $\mathcal{H}_j\rightarrow \mathcal{H}_k$                    $(\star\star\star)$} .
\end{center}
So if $\tilde{X}_{jk}\neq 0$ then from $(\star\star)$ and $(\star\star\star)$ we have that $\mathcal{H}_j$ and $\mathcal{H}_k$ $d_j=d_k$ and
\begin{equation}\label{eq:unitarnoscrowne bloki}
\tilde{X}_{jk} \tilde{X}_{jk}^{\dagger} = \tilde{X}_{jk}^{\dagger} \tilde{X}_{jk}=  ||\tilde{X}_{jk}||^2_2 \,\mathbb{I}_{d_j}.
\end{equation}
Therefore, if $d_j \neq d_k$, then $\tilde{X}_{jk} =0$ and the proof is completed.

\end{document}